\documentclass[letterpaper, 10pt, conference]{ieeeconf} 
\IEEEoverridecommandlockouts
\usepackage{graphicx}
 \usepackage{resizegather}
\usepackage{url}
\usepackage{xcolor}
\usepackage{amsmath}
\usepackage{mathtools}
 \mathtoolsset{showonlyrefs}
\usepackage{tikz}
\usepackage{verbatim}
\usepackage{subcaption}
\usepackage{pgfplots}
\pgfplotsset{compat=1.10}
\usepackage{graphicx}
\usepackage{epstopdf}
\usepackage{amssymb}
\usepackage{relsize}
\usepackage[textwidth=3.8em,textsize=scriptsize,disable]{todonotes}
\usepackage{algorithm}
\definecolor{Gray}{gray}{0.90}
\usepackage{colortbl}

\usepackage{algorithm}
\usepackage{algorithmicx}
\usepackage{algcompatible}
\usepackage{algpseudocode}
\usepackage{color}
\usepackage{blkarray}
\usepackage{sidecap}

\newtheorem{theorem}{\bf{Theorem}}[section]

\newtheorem{lem}[theorem]{\bf{Lemma}}

\newenvironment{definition}[1][Definition]{\begin{trivlist}
\item[\hskip \labelsep {\bfseries #1}]}{\end{trivlist}}

\usepackage{amssymb}
\usepackage{url}

\definecolor{amber}{rgb}{1.0, 0.75, 0.0}

\begin{document}

\title{\bf Priority Based Synchronization for Faster Learning in Games}
\author{ Abbasali Koochakzadeh and Yasin Yaz{\i}c{\i}o\u{g}lu
\thanks{Abbasali Koochakzadeh and Yasin~Yaz{\i}c{\i}o\u{g}lu are with the Department of Electrical and Computer Engineering at the University of Minnesota, Minneapolis, MN, USA. Emails:
Kooch002@umn.edu, ayasin@umn.edu}
}

\maketitle
\begin{abstract}

Learning in games has been widely used to solve many cooperative multi-agent problems such as coverage control, consensus, self-reconfiguration or vehicle-target assignment. One standard approach in this domain is to formulate the problem as a potential game and to use an algorithm such as log-linear learning to achieve the stochastic stability of globally optimal configurations. Standard versions of such learning algorithms are asynchronous, i.e., only one agent updates its action at each round of the learning process. To enable faster learning, we propose a synchronization strategy based on decentralized random prioritization of agents, which allows multiple agents to change their actions simultaneously when they do not affect each other's utility or feasible actions. We show that the proposed approach can be integrated into any standard asynchronous learning algorithm to improve the convergence speed while maintaining the limiting behavior (e.g., stochastically stable configurations).  We support our theoretical results with simulations in a coverage control scenario.

\end{abstract}


\section{Introduction}

Game theoretic formulations have been used to solve numerous multi-agent planning and control problems, both in cooperative and non-cooperative settings (e.g., \cite{marden2009cooperative,  yaziciouglu2021distributed,zhu2013distributed,wang2019game,bhattacharya2010game}). One such approach that has been extensively used for cooperative settings is to formulate the underlying coordination problem as a potential game so that every agent's utility function is aligned with the global objective function (e.g., \cite{marden2009cooperative}). The agents can then be driven to optimal configurations by iteratively revising their actions  in accordance with a suitable learning algorithm (e.g., see \cite{marden2009cooperative,Young93,jaleel2017transient,borowski2015fast,Blume93}, and the references therein). While these algorithms typically provide guarantees on the limiting behavior of the learning process, how the convergence time scales with the problem size (e.g., number of agents) depends on the type of game (e.g., \cite{jaleel2017transient,borowski2015fast, babichenko2016graphical,tatarenko2018learning}), and it may imply an impractically slow learning process in some cases.

In this paper, we propose a method that can be integrated into any standard asynchronous learning algorithm to facilitate faster convergence to the same limiting behavior by allowing multiple agents to update simultaneously as long as their next actions do not affect each other's utility. We particularly focus on the invariance of stochastically stable states under such a modification to the learning process since this is a common characterization of limiting behavior for potential games under stochastic, time-invariant learning algorithms (e.g., \cite{Blume93,Marden12,jaleel2017transient}). The stochastically stable states of an asynchronous algorithm can change when synchronous updates are allowed arbitrarily \cite{Marden12}. Synchronous algorithms where agents independently decide when to update their actions may lead to an even slower learning process to maintain stability (e.g., \cite{zhu2013distributed,Marden12,Yazicioglu17TCNS,wang2021stochastically}) or rely on strong assumptions such as the utility functions being always independent of the actions of other agents (e.g., \cite{hasanbeig2017synchronous}).  Alternatively, our proposed method is based on a decentralized random prioritization of agents to allow for synchronous updates only by uncoupled agents, i.e., agents that do not affect each other's utility or feasible actions at the current configuration. We consider a generic setting where the feasible actions of agents may be constrained by the current actions of themselves and other agents and the couplings among the agents may change during the learning process. We theoretically show the invariance of the stochastically stable states under the proposed synchronization method. We also numerically demonstrate how the proposed method improves the convergence speed of an asynchronous learning algorithm by considering a coverage control problem.

{\color{black}The organization of this paper is as follows: Section II includes some preliminaries on graph theory, stochastic stability, and game theory. Section III presents our proposed approach and main results. Section IV provides a coverage control problem as an example application and demonstrates the performance of the proposed approach via simulations. Finally, Section V concludes the paper. }

\section{Preliminaries}
\label{prelim}



\subsection{Graphs}
A graph, $G = (V, E)$, consists of a set of
nodes, $V$, and a set of edges, $E$, given by ordered pairs of nodes. Accordingly, any $(v,v')\in E$ denotes an edge from $v$ to $v'$ and the nodes $v$ and $v'$ are adjacent. The graph is said to be undirected if $(v,v')\in E$ implies $(v',v)\in E$.

A \emph{path} is a sequence of nodes such that an edge exists between any two consecutive nodes in the sequence. For any two nodes $v$ and $v'$, the \emph{distance} between the nodes $d(v,v')$ is the number of edges in a shortest path from $v$ to $v'$. We use the convention that the distance of a node to itself is zero. A graph is (strongly) \emph{connected} if the distance between any pair of nodes is finite.
A \emph{self-loop} is an edge between a vertex and itself.

Given a graph $G = (V, E)$, a \emph{spanning tree} rooted at some node $v\in V$, $T_v=(V,E'\subseteq E)$, is a subgraph of G such that there is a unique path on $T_v$ from any state $v' \neq v$ to $v$. For any $G=(V,E)$ and $v,v' \in V$, we use $G'=G+(v,v')$ to denote the graph $G'=(V,E\cup\{(v,v')\})$ and $G'=G-(v,v')$ to denote the graph $G'=(V,E\setminus \{(v,v')\})$.





\subsection{Stochastic Stability}




%
%
%

\begin{definition} \emph{(Regular Perturbed Markov Chain)}: Let $P_0$ be the transition matrix of a Markov chain over a finite state space $A$. A perturbed Markov chain with the noise parameter $\epsilon$ is called a regular perturbed Markov chain if
\begin{enumerate}
\item $P_{\epsilon}$ is aperiodic and irreducible for $\epsilon > 0$,
\item $\lim_{\epsilon \rightarrow 0} P_{\epsilon}=P_0$,
\item For any $a,a' \in A$ if $P_{\epsilon}(a,a')>0$, then there exists $R(a,a') > 0$ such that
\begin{equation}
\label{resdef_org}
0 < \lim_ {\epsilon \rightarrow 0^+} \frac{P_{\epsilon}(a, a')}{\epsilon^{R(a,a')}} < \infty,
\end{equation}
where $R(a,a')$ is called the resistance of the transition from $a$ to $a'$.
\end{enumerate}
\end{definition}

For simplicity in notation, we will use $R(\cdot)$ to denote the total resistance of $\cdot$, which may encode a set/sequence of feasible transitions in $P_\epsilon$ that will be clear from the context. 

\begin{definition} \emph{(Stochastically Stable State)}: Let $P_{\epsilon}$ denote a regular perturbed Markov chain over a state space, $A$. A state, $a\in A$, is stochastically stable if 
\begin{equation}
\label{sstab}
\lim_ {\epsilon \rightarrow 0^+} \mu^*_{\epsilon}(a) >0,
\end{equation}
where $\mu^*_{\epsilon}$ denotes the limiting distribution of $P_{\epsilon}$.
\end{definition}

The stochastically stable states of a regular perturbed Markov chain, $P_\epsilon$, are the recurrent states of the unperturbed chain, $P_0$, with the minimum stochastic potential \cite{Young93}. The stochastically stable states of any regular perturbed Markov chain can be characterized through a resistance tree analysis. For any state $a \in A$ of a regular perturbed  Markov chain, a spanning tree rooted at $a$ can be constructed as a directed graph $T_a$, where the nodes correspond to the states, the edges correspond to the feasible state transitions, and there is a unique directed path on $T_a$ from any state $a' \neq a$ to $a$. The resistance of such a tree, $R(T_a)$, is defined as the sum of the resistances of its edges, where the resistance of each edge is given as in (\ref{resdef}). A spanning tree $T_a$ is called a minimum resistance tree if  any spanning tree rooted at $a$ has at least as much resistance as $T_a$. The stochastic potential of a state, $a$, is the total resistance of its minimum resistance tree.

\subsection{Games}
A finite strategic game $\Gamma = (I, A, U)$ has three components: (1) a set of agents $I = \{1, 2,\hdots, n\}$, (2)
an action space $A = A_1 \times A_2 \times . . . \times A_n$, where each $A_i$ is the action set of agent $i$, and (3) a set of utility functions $U = \{U_1, U_2,\hdots, U_n\}$, where each $U_i : A \rightarrow \mathbb{R}$ is a mapping from the action space to real numbers. For any action profile $a \in A$, we use $a_{-i}$ to denote the
actions of agents other than $i$. Using this notation, an action profile $a$ can also be represented as $a = (a_i
, a_{-i})$. 

A class of games that is widely utilized in cooperative
control problems is the potential games, where the utilities of all agents are aligned with some global function over the action space (e.g., \cite{marden2009cooperative}). Constrained potential games (e.g., \cite{zhu2013distributed}) are a generalization, where the feasible actions of agents can be constrained by the current action profile.

\begin{definition} \emph{(Constrained Potential Game)}: A constrained game, $\Gamma=(I,A,U,C)$, has four components: 
\begin{enumerate}
\item A set of agents, $I = \{1, 2,\hdots, n\}$, 
\item An action space, $A = A_1 \times A_2 \times . . . \times A_n$, where each $A_i$ is the action set of agent $i$,
\item A set of utility functions $U = \{U_1, U_2,\hdots, U_n\}$, where each $U_i : A \rightarrow \mathbb{R}$ is a mapping from the action space to real numbers,
\item A set of constraint functions $C = \{C_1, C_2,\hdots, C_n\}$, where $C_i : A \rightarrow 2^{A_i}$ maps each action profile to the corresponding set of feasible actions of agent $i$, and $a_i \in C_i(a)$ for every $a\in A$. 
\end{enumerate}
Furthermore, the game is a constrained potential game if there exists a function, $\phi : A \rightarrow \mathbb{R}$, such that for every agent $i$, action profile $a \in A$, and feasible action $a_i^\prime \in C_i(a)$,
\begin{equation}
\label{cpotg}
U_{i}\left(a_{i}^{\prime}, a_{-i}\right)-U_{i}\left(a_{i}, a_{-i}\right)=\phi\left(a_{i}^{\prime}, a_{-i}\right)-\phi\left(a_{i}, a_{-i}\right).
\end{equation}
\end{definition}
Accordingly, unconstrained potential games are a special case, where $C_i(a)=A_i$ for every $i \in I$ and $a\in A$.

In game theoretic learning, the agents update their actions based on past observations in a repetitive play of the game. At each round, each agent $i \in  I$ plays an action $a_i$ and receives the utility $U_i(a)$. In a constrained setting, an updating agent chooses its next action from its feasible actions, $C_i(a)$, based on the utilities it may receive from such unilateral deviations. An unconstrained learning process is the special where $C_i(a)=A_i$ for all $a\in A$. We particularly focus on noisy best response type policies such as log-linear learning, Metropolis learning, and their variants. We denote the update policy followed by the agents as $\pi_\epsilon(C_i(a); U_i)$, which is a probability distribution over $C_i(a)$ that depends on  $U_i(a_i',a_{-i})$ for every $a_i' \in C_i(a)$. Here, $\epsilon \geq 0$ denotes the noise parameter. Accordingly, $\pi_0(C_i(a); U_i)$ is the best response policy that chooses an action $a'\in C_i(a)$ that maximizes $U_i(a_i',a_{-i})$ whereas  $\pi_\epsilon(C_i(a); U_i)$ denotes a noisy/perturbed version for any $\epsilon>0$. Such a standard asynchronous algorithm can be given as in Alg. \ref{alg:async}. 

\begin{algorithm}
{\small
\caption{A generic asynchronous learning algorithm}\label{alg:async}
\hspace*{0.5mm} \textbf{Input}: $\Gamma=(I,A,U,C)$, update policy $\pi_\epsilon$, $\epsilon \geq 0$ \\
\hspace*{0.5mm} \textbf{Initialization}: arbitrary  $a\in A$
\begin{algorithmic}[1]
\While{(1)}
\State Pick an agent $i$ uniformly at random.
\State Agent $i$ updates its action: 
 $$a_{i}^+\sim \pi_\epsilon(C_i(a); U_i).$$
\State Other agents maintain their actions: $a_{-i}^+=a_{-i}$.
\State $a=a^+$.
\EndWhile
\end{algorithmic}}
\end{algorithm}

\section{Main Results}
\label{T_Proofs}
In this section, we present our proposed approach and main theoretical results.

\subsection{Proposed Synchronization of Learning Algorithms}

Our proposed approach is mainly based on a distributed strategy that achieves a random prioritization of the agents and allows the simultaneous updates of agents whose local (constrained) action-updates do not influence each other's utility or feasible actions. 
\begin{definition} \emph{(Uncoupled Agents)}: 
\label{def:I_uii}
Consider any action profile $a \in A$, agent $i\in I$, and set of agents $J\subseteq I\setminus\{i\}$. Let $a$ be expressed as $a=(a_i,a_J,a_{-iJ})$, where $a_i$ is the action of agent $i$, $a_J$ denotes the actions of agents in $J$, and $a_{-iJ}$ denotes the actions of all the other agents. We say that agent $i$ is uncoupled from the agents in $J$ at action profile $a$ if
\begin{equation}
\label{U_i_eq1}
U_i(a_i',a_J',a_{-iJ})=U_i(a_i',a_J,a_{-iJ}), \; \forall \; {a_i'\in C_i(a), a_J'\in C_J(a)},\\
C_i(a_i',a_J',a_{-iJ})=C_i(a_i',a_J,a_{-iJ}), \; \forall \; {a_i'\in C_i(a), a_J'\in C_J(a)},
\end{equation}
where $C_J(a)=\prod_{j\in J}C_j(a)$ is the feasible joint actions for the agents in $J$. 
\end{definition}
Based on the definition above, it can be shown that if $i$ is uncoupled from $J$ at some $a\in A$, then $i$ is uncoupled from any $J' \subseteq J$ at $a$.
\begin{definition} \emph{(Coupling Functions)}: Consider any $\{I_1^c, I_2^c, \hdots, I_n^c\}$ such that each $I_i^c: A \mapsto 2^I$ is a mapping from the action space $A$ to the power set of the set of agents, $I$. We say that $\{I_1^c, I_2^c, \hdots, I_n^c\}$ is a set of valid coupling functions if the following conditions are all true: 
\begin{enumerate}
\item $i \in I_i^c(a)$, $ \forall i\in I,a\in A$. 
\item $i \in I_j^c(a) \Leftrightarrow j \in I_i^c(a)$, $\forall i,j \in I, a\in A$.
\item $i$ is uncoupled from $I \setminus I^c_i(a)$ at $a$, $\forall i\in I,a\in A$.
\end{enumerate}
\end{definition}

Our proposed synchronization of any standard asynchronous learning algorithm as in Alg. \ref{alg:async} is provided in Alg. \ref{alg:sync}. Here, each agent $i\in I$ independently chooses an intended action $\bar{a}_i$ as per $\pi_\epsilon$ in Alg. \ref{alg:async}. It then determines its priority variable, $\beta_i$, which is set to zero if agent $i$ plans to repeat its action as per the previous step ($\bar{a}_i = a_i$) or due to its inertia (with probability $\kappa$).  Otherwise, $\beta_i$ is picked uniformly at random from the interval $[0,1]$, i.e., $\beta_i \sim \mathcal{U}(0,1)$. If agent $i$ plans to change its action ($\bar{a}_i \neq a_i$, $\kappa_i>\kappa$), it is allowed to do so as long as it has the unique highest priority in $I^c_i(a)$. Agent $i$ can check if this condition holds by gathering the priority variables of all other agents in $I^c_i(a)$ via communications. As such, each agent determines its desired next action independently, and coupled agents communicate with each other to determine who should be allowed to move (change its action) based on their priority variables. This process ensures that any moving agent is uncoupled from the others moving simultaneously. 

One important parameter in the execution of Alg. \ref{alg:sync} is the selection of coupling functions $\{I_1^c, I_2^c, \hdots, I_n^c\}$, which in general is not unique. Since any moving agent $i$ restricts the others in $I_i^c(a)$ to remain stationary, it would be desired to find coupling functions that are minimally restrictive, i.e., each $I_i^c(a)$ is as small as possible.  As we will also show with an example in Section \ref{example}, such minimally restrictive maps can be determined in many scenarios based on the problem specifications. For example, for a team of mobile robots where each agent's utility and feasible actions depend only on the other agents nearby, each $I_i^c(a)$ may be defined as the agents that are sufficiently close to $i$. When no such prior information is available, one valid, yet conservative, selection is to set $I_i^c(a)=I$ for every $i\in A$ and $a\in A$, in which case Alg. \ref{alg:sync} would not allow multiple agents to move simultaneously.

\begin{algorithm}{\small
\caption{Proposed synchronization of a learning alg.}\label{alg:sync}
\hspace*{0.5mm} \textbf{Input}: {$\Gamma=(I,A,U,C)$, policy $\pi_\epsilon$, $\epsilon \geq 0$, inertia $\kappa \in (0,1)$} \\
\hspace*{0.5mm} \textbf{Initialization}: arbitrary $a\in A$
\begin{algorithmic}[1]
\While{(1)}
\State Each agent $i$ executes the following:
\State Generate a random $\kappa_i \sim \mathcal{U}(0,1)$.
\State Choose an intended action $\bar{a}_i$: $$\bar{a}_i\sim \pi_{\epsilon}(C_i(a); U_i).$$
\State Generate the priority variable $\beta_i$:
$$\beta_i=\left\{\begin{array}{ll} \mbox{$0$, if $\bar{a}_i =a_i$ or $\kappa_i \leq \kappa$,} \\ \mbox{$\sim \mathcal{U}(0,1)$, otherwise.}\end{array}\right.$$
\State Update the action: 
$$a_i^+=\left\{\begin{array}{ll} \mbox{$a_i$, if $\beta_i=0$ or $\exists j \neq i \in I_i^{c}(a): \beta_j \geq \beta_i$,} \\ \mbox{$\bar{a}_i$, otherwise.}\end{array}\right.$$
\State $a_i=a^+_i$.
\EndWhile
\end{algorithmic}}
\end{algorithm}




\subsection{Relating the Transition Probabilities under Algs. \ref{alg:async} and \ref{alg:sync}}

In the remainder of this section, we will focus on how the proposed synchronization affects the limiting behavior for asynchronous algorithms that induce a regular perturbed Markov chain over the action space. Prior to presenting our main results on the impact of the proposed synchronization on the stochastically stable states, we first relate the transition probabilities under Alg. \ref{alg:sync}, $P'_\epsilon$, to the transition probabilities under Alg. \ref{alg:async}, $P_\epsilon$. For any $a,a'\in A$, let $K_{a,a'}\subseteq I$ be the set of agents whose actions are different in $a$ and $a'$, i.e.,
\begin{equation}
\label{K}
K_{a,a'}=\{i\in I\mid a_i\neq a_i'\}.
\end{equation}
At any iteration of Alg. \ref{alg:sync}, where the current action profile is $a$, let $S_a \subseteq I$ be the set of agents who are allowed to simultaneously change their actions, i.e.
\begin{equation}
\label{S-set}
S_a= \{i \in I \mid a^+_i = \bar{a}_i \neq a_i\}.
\end{equation}

Note that $S_a$ is determined by the couplings encoded in $I_1^c(a), \hdots,I_n^c(a)$ and the random variables $\bar{a}$ , $\kappa_1, \hdots, \kappa_n$, and $\beta_1, \hdots, \beta_n$. We use $Pr(S_a=K_{a,a'}\mid \bar{a})$ to denote the probability of observing $S_a$ in \eqref{S-set} being equal to $K_{a,a'}$ in \eqref{K}, given the intended actions of agents, $\bar{a}$,  in line 4 of Alg. \ref{alg:sync}. 
Let $\bar{A}_{a,a'}$ be the set of $\bar{a}$ that allows for a transition from $a$ to $a'$, i.e.,
\begin{equation}
    \label{A_bar}
    \bar{A}_{a,a'}= \{\bar{a} \in A \mid \bar{a}_i = a_i', \; \forall{i \in K_{a,a'}} \}.
\end{equation}
 Then, the probability of switching from $a$ to $a'$ under Alg. \ref{alg:sync} can be expressed as
\begin{equation}
\label{main_equation1}
P'_\epsilon(a,a')=\sum_{\bar{a}\in\bar{A}_{a,a'}} Pr(S_a=K_{a,a'}\mid \bar{a}) Pr(\bar{a}),
\end{equation}
where $Pr(\bar{a})$ is the probability that $\pi_\epsilon$ produces $\bar{a}_i$ for every agent $i\in I$ in line 4 of the Alg. \ref{alg:sync} when the current action profile is $a$. Given $\bar{a}$, the following lines in Alg. \ref{alg:sync} leads to $a^+$ based on the values of priority variables and the couplings among the agents. Action profiles $\bar{a}$ and $a^+$ together determine the set $S_a$. Accordingly, $Pr(S_a=K_{a,a'}\mid \bar{a})$ denotes the probability of observing $S_a=K_{a,a'}$ given $\bar{a}$.  Since $Pr(\bar{a})$ equals the product of probabilities that each agent chooses $\bar{a}_i$ by following $\pi_\epsilon$, \eqref{main_equation1} can be expressed in terms of the transition probabilities under Alg. \ref{alg:async} as

\begin{equation}
\label{main_equationn1}
P'_\epsilon(a,a')=\sum_{\bar{a}\in\bar{A}_{a,a'}} Pr(S_a=K_{a,a'}\mid \bar{a}) \prod_{j\in I} nP_\epsilon(a,\alpha^{j}),
\end{equation}
where each $\alpha^j$ is the state whose entries are
\begin{equation}
\label{alpha_j}
\alpha^{j}_i = \left\{\begin{array}{ll} \mbox{$a_i$, if $i \neq j$,} \\ \mbox{$\bar{a}_i$, if $i=j$.}\end{array}\right.
\end{equation}

In \eqref{main_equationn1}, each term $nP_\epsilon(a,\alpha^{j})$ denotes the probability of agent $j$ choosing $\bar{a}_j$ by following $\pi_\epsilon$. Here, multiplying by $n$ inverts the multiplier $1/n$ in  $P_\epsilon(a,\alpha^{j})$ due to the probability of randomly picking $j$ in line 2 of Alg. \ref{alg:async}.


\subsection{Impact on the  Stochastically Stable States}
\label{section_IV.B.}


Our main result in this section, Theorem \ref{thm-stoch}, shows that the proposed synchronization in Alg. \ref{alg:sync} does not change the stochastically stable states induced by an asynchronous algorithm in Alg. \ref{alg:async}. To this end, we first present some lemmas.  Our first result, Lemma \ref{resistance_for_probs_sum_multiplication}, will later be used when expressing the resistances of transition probabilities given as products and sums of other probabilities as in \eqref{main_equation1}.

\begin{lem}
\label{resistance_for_probs_sum_multiplication}
Let $p_{1}(\epsilon), p_{2}(\epsilon), \hdots, p_{n}(\epsilon)$ be functions of $\epsilon$ such that for each $p_i(\epsilon)$, there exists some $r_i\geq0$ satisfying
\begin{equation}
\label{resdef}
0 < \lim_ {\epsilon \rightarrow 0^+} \frac{p_i(\epsilon)}{\epsilon^{r_i}} < \infty.
\end{equation}
Then, the following equations are satisfied:
\begin{equation}
\label{resprod}
0 < \lim_{\epsilon \rightarrow 0^+} \frac{\prod_{i=1}^{n} p_{i}(\epsilon)}{\epsilon^{r_p}} < \infty,
\end{equation}
\begin{equation}
\label{ressum}
0 < \lim_{\epsilon \rightarrow 0^+} \frac{\sum_{i=1}^{n} p_{i}(\epsilon)}{\epsilon^{r_s}} < \infty,
\end{equation}
 where $r_p = \sum_{i=1}^{n} r_i$ and $r_s = \min(r_1, r_2, \hdots, r_n)$.
\end{lem}

\begin{proof}
Since every limit in \eqref{resdef} is positive and finite, the product of those limits are also positive and finite, i.e., 
\begin{equation}
\label{resprod2}
0 <  \prod_{i=1}^{n}\lim_{\epsilon \rightarrow 0^+} \frac{p_{i}(\epsilon)}{\epsilon^{r_i}} < \infty.
\end{equation}
Since all the limits in \eqref{resprod2} are taken as $\epsilon$ goes down to zero, the product of limits in \eqref{resprod2} equals the limit of products, i.e.,
\begin{equation}
\label{resprod3}
\prod_{i=1}^{n}\lim_{\epsilon \rightarrow 0^+} \frac{p_{i}(\epsilon)}{\epsilon^{r_i}} = \lim_{\epsilon \rightarrow 0^+} \frac{\prod_{i=1}^{n}p_{i}(\epsilon)}{\prod_{i=1}^{n}\epsilon^{r_i}}= \lim_{\epsilon \rightarrow 0^+} \frac{\prod_{i=1}^{n}p_{i}(\epsilon)}{\epsilon^{r_p}},
\end{equation}
where $r_p = \sum_{i=1}^{n} r_i$. Using \eqref{resprod2} and \eqref{resprod3}, we obtain \eqref{resprod}.

Next, we prove that \eqref{ressum} holds. For each $r_i$  and $r_s = \min(r_1, r_2, \hdots, r_n)$, we have
\begin{equation}
\label{R_i_resdef}
\lim_ {\epsilon \rightarrow 0^+} \frac{\epsilon^{r_i}}{\epsilon^{r_s}}= \left\{\begin{array}{ll} \mbox{$0$, if $r_i>r_s$,} \\ \mbox{$1$, if $r_i=r_s$.}\end{array}\right.
\end{equation}
Using\eqref{resdef} and \eqref{R_i_resdef}, we obtain
\begin{equation}
\label{R_i_ej_resdef}
\lim_ {\epsilon \rightarrow 0^+} \frac{p_{i}(\epsilon)}{\epsilon^{r_s}}=\lim_ {\epsilon \rightarrow 0^+} \frac{p_{i}(\epsilon)}{\epsilon^{r_i}}\frac{\epsilon^{r_i}}{\epsilon^{r_s}}= \left\{\begin{array}{ll} \mbox{$0$, if $r_i>r_s$,} \\ \mbox{$\lim \limits_{\epsilon \rightarrow 0^+} \frac{p_{i}(\epsilon)}{\epsilon^{r_i}}$, if $r_i=r_s$.}\end{array}\right.
\end{equation}
Since $r_i=r_s$ for at least one $i\in \{1, \hdots,n \}$, \eqref{resdef} and \eqref{R_i_ej_resdef} together imply \eqref{ressum}.

\end{proof}

We next show that any feasible transition $(a,a'\neq a)$ of the asynchronous learning ($P_\epsilon$) is also feasible under the proposed synchronous version ($P_\epsilon'$) and has equal resistance in these two regular perturbed Markov chains.

\begin{lem}
\label{first_condition}
Let $P_\epsilon$ be a regular perturbed Markov chain induced by an asynchronous learning algorithm as in Alg. \ref{alg:async} and let $P_\epsilon'$ be the chain induced by its proposed synchronous version as in Alg. \ref{alg:sync}. Any feasible transition $(a,a'\neq a)$ in $P_\epsilon$ is also feasible in $P_\epsilon'$ and has equal resistances in those two Markov chains.
\end{lem}
\begin{proof}
Consider any feasible transition of $P_\epsilon (a,a')>0$ from some state $a$ to some other state $a'\neq a$. Since $P_\epsilon$ is induced by an asynchronous learning algorithm, $a$ and $a'$ must differ only in one agent's action, say the $k^{th}$ dimension, i.e., $K_{a,a'}=\{k\}$ in \eqref{K}. For the same transition in the synchronous version, we can express $P'_\epsilon(a,a')$ as given in \eqref{main_equationn1}. Let $\bar{a}^j$ be an intended action profile whose $j^{th}$ entry is equal $a'$ and other entries are from $a$, i.e.
\begin{equation}
\label{bar_a^j}
\bar{a}^{j}_i = \left\{\begin{array}{ll} \mbox{$a_i$, if $i \neq j$,} \\ \mbox{$a_i'$, if $i=j$.}\end{array}\right.
\end{equation}
Note that since $a$ and $a'$ only different in $k^{th}$ entry then based on \eqref{bar_a^j}, 
$\bar{a}^k_k = a'_k$. Now based on \eqref{alpha_j}, if we compute $\alpha^k$ given an $\bar{a}=\bar{a}^k$ we have $\alpha^k_k=a'_k$ and other entries of $\alpha^k$ are same as $a$. Hence, $\alpha^k = a'$. Furthermore based on \eqref{A_bar}, we have $\alpha^k \in \bar{A}_{a,a'}$. If we separate $\bar{a}^k$ from the rest of the summation in \eqref{main_equationn1}, we obtain $P'_\epsilon(a,a')$ as
\begin{equation}
\label{main_equation3}
P_r(S_a=\{k\}\mid \bar{a}^{k})nP_\epsilon(a,\alpha^{k})+\sum_{\bar{a}\in\bar{A}_{a,a'}\setminus \{\bar{a}^k\}} P_r(S_a=\{k\}\mid\bar{a})\prod_{j\in I}nP_\epsilon(a,\alpha^j),
\end{equation}
where the states $\alpha^j$ are computed from $a$ and corresponding $\bar{a}\in\bar{A}_{a,a'}\setminus \{\bar{a}^k\}$ as in \eqref{alpha_j}.
Note that $Pr(S_a=\{k\}\mid \bar{a}^k)$ is guaranteed to be bounded away from zero (independent from $\epsilon$). More specifically,
\begin{equation}
  \label{Prs}  
  Pr(S_a=\{k\}\mid \bar{a}^k)\geq (1-\kappa)\kappa^{n-1},
\end{equation}
where the lower bound is the probability that every agent except $k$ stays stationary due to inertia (line 5 in Alg. \ref{alg:sync}, $\kappa_i\leq \kappa$), which is one feasible way that always results in $S_a=\{k\}$ when $\bar{a}=\bar{a}^k$.
Since $P_\epsilon (a,a')>0$ and $P_\epsilon (a,\alpha^k)=P_\epsilon (a,a')$, we can use \eqref{main_equation3} to conclude that $P'_\epsilon(a,a')>0$, i.e., all feasible transitions of the asynchronous learning are feasible in the synchronous version as well. 

Next we will show that the resistances are also preserved during the synchronization process proposed in Alg. \ref{alg:sync}. Note that any $a,a'$ such that $P_\epsilon (a,a')>0$ only differ in one entry $k$. Accordingly, as per \eqref{alpha_j}, $\alpha^k=a'$ for any $\bar{a}\in\bar{A}_{a,a'}\setminus \{\bar{a}^k\}$. Hence, we have $P_\epsilon(a,a')=P_\epsilon(a,\alpha^k)$ as a multiplayer in every summand of \eqref{main_equation3}. Accordingly, using Lemma \ref{resistance_for_probs_sum_multiplication} and the fact that $P_r(S_a=\{k\}\mid \bar{a}^k)$ is bounded away from zero independent of the noise parameter $\epsilon$ (hence it does not affect the resistance) as per \eqref{Prs}, we obtain
\begin{equation}
\label{P''_res2}
R'(a,a')=R(a,a'),
\end{equation}
where $R'(a,a')$ is the resistance of the transition $(a,a')$ on $P'_\epsilon$ and $R(a,a')$ is the resistance of the transition $(a,a')$ on $P_\epsilon$. Consequently, we conclude that any feasible transition $(a,a'\neq a)$ in $P_\epsilon$ is also feasible in $P_\epsilon'$ and has equal resistances in $P_\epsilon$ and $P'_\epsilon$.
\end{proof}

Next, we show that the unperturbed Markov chains $P_0$ and $P'_0$ have the same set of recurrent states.
\begin{lem}
\label{P_0,P'_0}
Let $P_0$ be an unperturbed Markov chain induced by an asynchronous learning algorithm as in Alg. \ref{alg:async} and let $P'_0$ be the unperturbed chain induced by its proposed synchronous version as in Alg. \ref{alg:sync}. Then, $P_0$ and $P'_0$ have the same set of recurrent states.
\end{lem}
\begin{proof}
For any agent $i\in I$, under Alg. \ref{alg:sync} there is a non-zero probability that $\kappa_i > \kappa$ and $\kappa_j \leq \kappa$ for all $j\neq i$, in which case $i$ would make an asynchronous update just as in Alg. \ref{alg:async}. Hence, $P'_0$ has all the feasible transitions of $P_0$. Any additional feasible transition in $P'_0$ involves a synchronous update by agents who are uncoupled with each other (line 6 of Alg. \ref{alg:sync}), i.e., they do not influence each other's utility or feasible actions when deviating from $a$ as per \eqref{U_i_eq1} in uncoupled agents definition. Accordingly, for any $a,a'\in A$ such that $P'_0(a,a')>0$ and $P_0(a,a')=0$, there is a multi-step transition from $a$ to $a'$ on $P_0$, i.e., the respective agents switch their actions in $a$ to those in $a'$ sequentially. Hence, for any $a,a' \in A$, $a'$ is reachable from $a$ on $P'_0$ if and only if $a'$ is reachable from $a$ on $P_0$. Since the recurrent states are determined by the reachability among the states in $A$, $P_0$ and $P'_0$ have the same set of recurrent states.
\end{proof}

Finally, we will show that the regular perturbed Markov chains induced by Alg. \ref{alg:async} and Alg. \ref{alg:sync} have the same stochastically stable states.

\begin{theorem}
\label{thm-stoch}
Let $P_\epsilon$ be a regular perturbed Markov chain induced by an asynchronous learning algorithm as in Alg. \ref{alg:async} and let $P_\epsilon'$ be the regular perturbed Markov chain induced by its proposed synchronous version as in Alg. \ref{alg:sync}. Then, $P_\epsilon$ and $P_\epsilon'$ have the same stochastically stable states.
\end{theorem}


\begin{proof} 
The stochastic potential on any state $\alpha$ is determined by the the total resistance of the minimum resistance trees rooted at $\alpha$, say $R(T_\alpha)$ on $P_\epsilon$ and $R'(T_\alpha')$ on $P_\epsilon'$. In light of Lemma \ref{first_condition}, any tree on $P_\epsilon$ is also feasible and has the same total resistance on $P_\epsilon'$, which implies $R'(T_\alpha')\leq R(T_\alpha)$. What we will show here is that $ R(T_\alpha)\leq R'(T_\alpha')$ is also true, which implies that $ R(T_\alpha)= R'(T_\alpha')$ and each state has equal stochastic potential in $P_\epsilon$ and $P'_\epsilon$. 


Let $h'=\{\bar{a}^0,\bar{a}^1,\hdots,\bar{a}^M\}$ be a minimum resistance path on the synchronous version between arbitrary states $\bar{a}^0$ and $\bar{a}^M$. We will show that for any such $h'$, there is a path, $h$, in the asynchronous version from $\bar{a}^0$ to $\bar{a}^M$ such that $R(h)=R'(h')$, i.e., the paths have equal total resistance. Let $P'_\epsilon(h')$ be the probability that the synchronous algorithm takes the system from $\bar{a}^0$ to $\bar{a}^M$ via $h'$, i.e., 
\begin{equation}
\label{p_path}
    P'_\epsilon(h')=\prod_{j=0}^{M-1} P'_\epsilon(\bar{a}^j,\bar{a}^{j+1}).
\end{equation} 
For any pair $(\bar{a}^j,\bar{a}^{j+1})$ on $h'$, let {\small{ $K_{\bar{a}^j,\bar{a}^{j+1}}=\{k_1, k_2, \hdots, k_m\}$}} be the set of agents whose actions are different in $\bar{a}^j$ and $\bar{a}^{j+1}$ as given in \eqref{K}. 

Furthermore, for every $k_i\in  K_{\bar{a}^j,\bar{a}^{j+1}}$, let  $\alpha^{k_i}$ be the state as given in \eqref{alpha_j} where we set $a=\bar{a}^j$ and $a'=\bar{a}^{j+1}$.  Let $R'(\bar{a}^j,\bar{a}^{j+1})$ be the resistance of $P_\epsilon'(\bar{a}^j,\bar{a}^{j+1})$. Using \eqref{main_equation1}, \eqref{p_path}, and Lemma \ref{resistance_for_probs_sum_multiplication}, it can be shown that
\begin{equation}
\label{fixed_j_R'}
R'(h')= \sum_{j=0}^{M-1} R'(\bar{a}^j,\bar{a}^{j+1}) =\sum_{j=0}^{M-1} \sum_{i=1}^m R(\bar{a}^j,\alpha^{k_i}),
\end{equation}
where $R(\bar{a}^j,\alpha^{k})$ is the resistance of the transition from $\bar{a}^j$ to $\alpha^{k_i}$ under the asynchronous learning algorithm. 
Moreover, for any $(\bar{a}^j,\bar{a}^{j+1})$ in $h'$, we can construct a path on the asynchronous version $h=\{h_0, h_1, \hdots, h_{M-1}\}$, where  $h_j=\{b^1,b^2,\hdots,b^{m+1}\}$ such that $b^1=\bar{a}^j$,$b^{m+1}=\bar{a}^{j+1}$, and each $b^{i+1}$ is the state obtained from $b^{i}$ when agent $k_i\in K_{\bar{a}^j,\bar{a}^{j+1}}$ unilaterally changes its action from $b^{i}_{k_i}$ to $\bar{a}^{j+1}_{k_i}$. Accordingly, for any $i \in \{2,\hdots,m+1\}$,
\begin{equation}
\label{b_ieq}
b^i_k = \left\{\begin{array}{ll} \mbox{$\bar{a}^{j+1}_k$ , if $k \in \{k_1, \hdots, k_{i-1}\}$} \\ \mbox{$\bar{a}^j_k$, otherwise.}\end{array}\right.
\end{equation}

Note that for any feasible transition $(\bar{a}^j,\bar{a}^{j+1})$ of the synchronous version, due to line 6 in Alg. \ref{alg:sync}, the agents in $K_{\bar{a}^j,\bar{a}^{j+1}}$ are all uncoupled from each other when the system is at $\bar{a}^j$. Accordingly, for each $k_i \in K_{\bar{a}^j,\bar{a}^{j+1}}$, the transitions made by the other agents along $h_j$ have no impact on $k_i$'s set of feasible actions or utilities from those actions. Accordingly, all the feasible actions and the corresponding utilities are the same for agent $k_i$ when it is allowed to update its action at $b^i$ or $\bar{a}^{j}$, i.e.,  
\begin{equation}
\label{prob-policy-equality1}
C_{k_i}(b^{i})=C_{k_i}(\bar{a}^j),
\\
U_{k_i}(a_i', b^{i}_{-k_i})=U_{k_i}(a_i',\bar{a}^{j}_{-k_i}), \forall a_i' \in C_{k_i}(\bar{a}^j).
\end{equation}
Hence, the probability of the transition $(b^{i},b^{i+1})$ is equal to the probability of transition $(\bar{a}^j,\alpha^{k_i})$ due to \eqref{prob-policy-equality1}, i.e.,
\begin{equation}
\label{p_equality}
P_\epsilon(b^{i},b^{i+1}) = P_\epsilon(\bar{a}^j,\alpha^{k_i}),
\end{equation}
where $\alpha^{k_i}$ be the state as defined in \eqref{alpha_j} for $a=\bar{a}^j$ and $a'=\bar{a}^{j+1}$. Furthermore,  the probability of traversing the path $h_j$ under the asynchronous algorithm is
\begin{equation}
\label{Ph_j}
P_\epsilon(h_j)=\prod_{i=1}^{m} P_\epsilon(b^i,b^{i+1}).
\end{equation}
Using \eqref{p_equality}, \eqref{Ph_j}, and Lemma \ref{resistance_for_probs_sum_multiplication}, we obtain 
\begin{equation}
\label{Rh_j}
R(h_j)=\sum_{i=1}^{m} R(\bar{a}^j,\alpha^{k_i}),
\end{equation}
\begin{equation}
\label{Rh}
R(h)=\sum_{j=0}^{M-1} R(h_j) = \sum_{j=0}^{M-1}\sum_{i=1}^{m} R(\bar{a}^j,\alpha^{k_i}).
\end{equation}
Based on \eqref{fixed_j_R'} and \eqref{Rh}, we obtain  $R(h)=R'(h')$. Accordingly, the total resistance of a minimum resistance path
 from $\bar{a}^0$ to $\bar{a}^M$ in $P_\epsilon$ is at most $R'(h')$. Since this is true for the minimum resistance paths between any pair of agents, for any $\alpha \in A$ the minimum resistance trees rooted at $\alpha$ satisfy $ R(T_\alpha)\leq R'(T_\alpha')$. Since $ R'(T_\alpha')\leq R(T_\alpha)$ is also true, we conclude that any state $\alpha \in A$ has the same stochastic potential in $P_\epsilon$ and $P'_\epsilon$. 
 Moreover, due to Lemma \ref{P_0,P'_0}, $P'_0$ and $P_0$ have the same recurrent states. Since the stochastically stable states of a regular perturbed Markov chain are the recurrent states of the unperturbed chain with the minimum stochastic potential \cite{Young93}, we conclude that $P_\epsilon$ and $P'_\epsilon$ have the same stochastically stable states. 

\end{proof}

\section{Example: Distributed Graph Coverage}
\label{example}
We numerically demonstrate the performance of proposed synchronization method in a coverage control problem over a graph \cite{Yazicioglu17TCNS}. 
In this distributed graph coverage problem, $n$ mobile agents start with an arbitrary deployment on a connected, undirected graph $G = (V,E)$.  Let $a_i(t) \in V$ denote the position of agent $i$ on $G$ at time $t$. Suppose that each agent can cover the nodes within a distance of one from its current position. Accordingly, the set of nodes covered by agent $i$ at time $t$ is

\begin{equation}
\label{Cov_i}
Cov_i(t)=\{v\in V \mid d(v,a_i(t)) \leq 1\},
\end{equation}
where $d(v,a_i(t))$ denotes the distance between node $v$ and agent $i$'s position $a_i(t)$. Then, the set of covered nodes at time $t$ is the union of the sets of nodes covered by the agents:
\begin{equation}
\label{Coverage_set_def}
Cov(t)=\bigcup_{i=1}^{n} Cov_i(t).
\end{equation}
The objective of the distributed graph coverage problem is to drive the agents to an optimal configuration that maximizes
\begin{equation}
\label{Coverage_func}
\phi(a(t))=\sum_{v \in Cov(t)} \omega (v),
\end{equation}
where $\omega (v) \geq 0$ denotes the value of node $v$. In this setting, multiple agents may be present at the same node, and each agent can either maintain its position or move to an adjacent node in the next time step. Accordingly, for each agent $i$, the next action belongs to the constrained set
\begin{equation}
\label{Cov_cons}
C_i(a(t))=\{v\in V \mid d(v,a_i(t)) \leq 1\}.
\end{equation}
As shown in \cite{Yazicioglu17TCNS}, such a distributed graph coverage problem can be formulated as a constrained potential game by defining each agent's utility as the total value of nodes covered only by itself, i.e.,
\begin{equation}
\label{U_i}
U_i(a(t)) = \sum_{v \in Cov_i(t)\setminus Cov(t)} \omega (v).
\end{equation}
For the resulting game, an asynchronous learning algorithm such as the binary log-linear learning (BLLL) \cite{Marden12} can be used to make the global maximizers of \eqref{Coverage_func} stochastically stable \cite{Yazicioglu17TCNS}. In this section, we provide simulation results to compare  BLLL and its synchronous version in Alg. \ref{alg:syn_BLLL}, which is obtained by using our proposed approach in Alg. \ref{alg:sync}.

\begin{algorithm}{\small
\caption{Synchronous version of BLLL algorithm}\label{alg:syn_BLLL}
\hspace*{0.5mm} \textbf{Input}: $\Gamma=(I,A,U,C)$, $\epsilon>0$ (small), inertia $\kappa \in (0,1)$, \\
\hspace*{0.5mm} \textbf{Initialization}: $a\in A$
\begin{algorithmic}[1]
\While{(1)}
\State Each agent $i$ executes the following:
    \State Generate a random $\kappa_i
\sim \mathcal{U}(0,1)$.
\State  Randomly pick an alternative action $a_i'\in C_i(a)$.
\State  Choose an intended action $\bar{a}_i$:
$$\bar{a}_i=\left\{\begin{array}{ll} \mbox{$a_i$, w.p.
{\color{black}$\frac{e^{U_{i}\left(a_i, a_{-i}\right) / \epsilon}}{e^{U_{i}\left(a_i, a_{-i}\right) / \epsilon}+e^{U_{i}(a_i',a_{-i}) / \epsilon}}$,}} \\ 
\mbox{$a_i'$, otherwise.}\end{array}\right.$$
\State Generate the priority variable $\beta_i$:
$$\beta_i=\left\{\begin{array}{ll} \mbox{$0$, if $\bar{a}_i =a_i$ or  $\kappa_i \leq \kappa$,} \\ 
\mbox{$\sim \mathcal{U}(0,1)$, otherwise.}\end{array}\right.$$
\State Update the action: 
$$a_i^+=\left\{\begin{array}{ll} \mbox{$a_i$, if $\beta_i=0$ or $\exists j \neq i \in I_i^{c}(a): \beta_j \geq \beta_i$,} \\ \mbox{$\bar{a}_i$, otherwise.}\end{array}\right.$$
\State $a_i=a^+_i$.
\EndWhile
\end{algorithmic}}
\end{algorithm}

In light of \eqref{Cov_cons}, the agents do not affect each other's constrained action sets in the graph coverage game. Furthermore, an agent's utility in the next time step cannot be affected by the actions of others that are at least 5 hops away on the graph. This is because when two agents have a distance of five or more, their sets of covered nodes in \eqref{Cov_i} cannot intersect in the next time step no matter how they move. Accordingly, we pick the coupling functions $\{I_1^c, \hdots, I_n^c\}$ as
\begin{equation}
\label{I^c_4hub}
I_i^{c}(a(t))=\{j \in I| \; d(a_i(t),a_j(t))\leq 4\}.
\end{equation}

\subsection{Numerical Results}
We demonstrate the performance of BLLL and its proposed synchronous version in the graph coverage game on a grid environment as shown in Fig. \ref{fig:positions}. This environment has some static obstacles and 80 feasible nodes. Each node has at most 4 neighbors (up, down, left, right) and has a value assigned from $\{1,3,5,7\}$. The value of each node is illustrated via its size in Fig. \ref{fig:positions}, where bigger nodes have higher value. We consider 5 mobile agents, for which a globally optimal configuration is also shown in Fig. \ref{fig:positions} by coloring the corresponding locations of agents in blue. Such an allocation makes the total value of covered nodes equal to 106 in this example.

\begin{figure}[htb]
    \centering
    \includegraphics[width=0.2\textwidth]{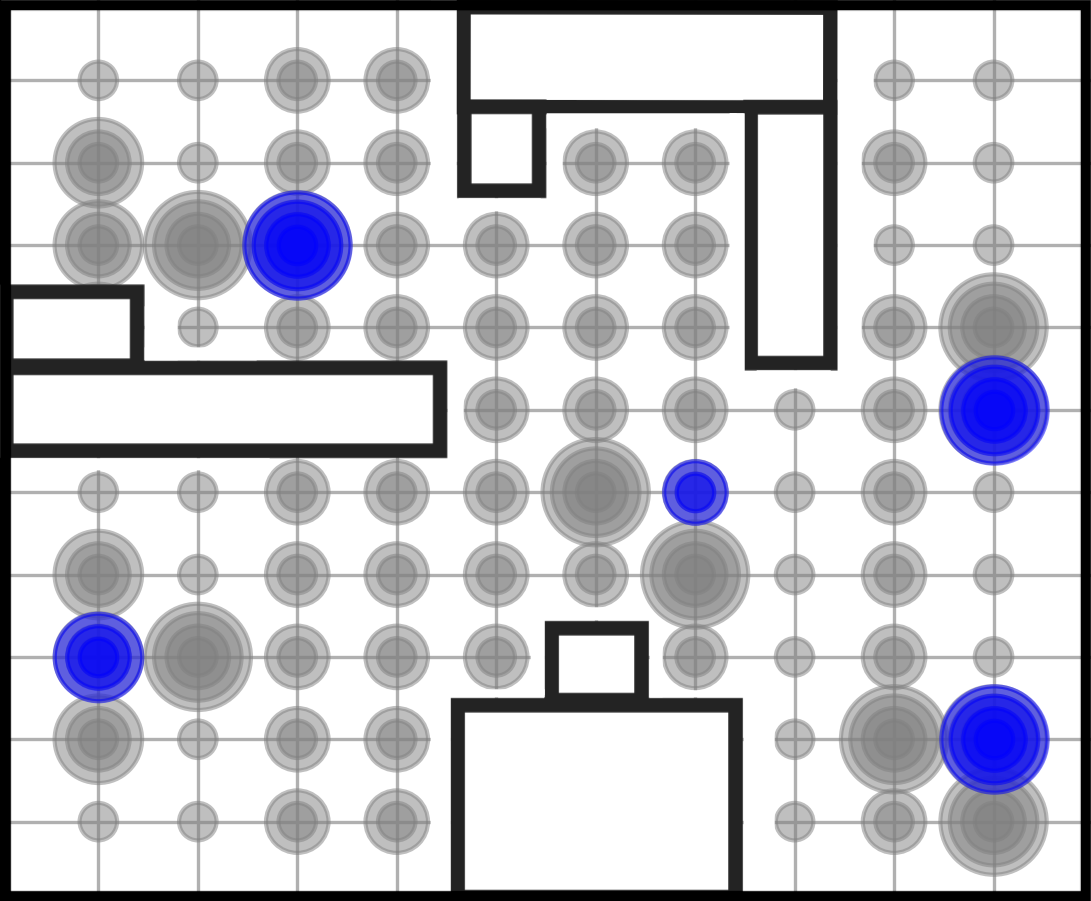}
    \caption{ A grid environment with obstacles (white boxes) and 80 feasible nodes. Size of each node denotes its value. A globally optimal allocation of 5 agents is highlighted in blue.}
    \label{fig:positions}
\end{figure}

 In the simulations, we use a noise parameter of $\epsilon=0.4$ for both BLLL and its proposed synchronous version. In the synchronous version, we use $\kappa = 0.01$ as the inertia parameter. To obtain a statistically significant comparison of the two algorithms, we randomly pick 50 initial configurations, $a(0) \in A$, for the agents. For each initial configuration, we run both learning algorithms separately over a horizon of 4000 rounds. Accordingly, we obtain 50 random runs for each algorithm. In Fig. \ref{fig:plots}, we report the time evolution of the average, minimum, and maximum values of the global objective in \eqref{Coverage_func} under each algorithm based on these runs. In Fig. \ref{fig:plots2}, we show the evolution of the global objective under the two algorithms in one of these 50 cases as an example.

\begin{figure}[htb!]
\includegraphics[trim =20mm 0mm 0mm 0mm, clip, width=0.52\textwidth]{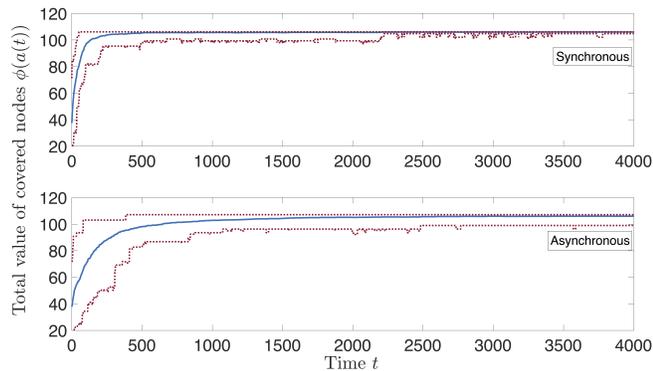}
\caption{{Evolution of the average (solid) and the minimum/maximum (dashed) values of the total value of covered nodes based on the 50 randomly initialized runs of BLLL (bottom) and its proposed synchronous version (top).  
}}
\label{fig:plots}
\end{figure}

\begin{figure}[htb!]
\includegraphics[trim =20mm 0mm 0mm 0mm, clip, width=0.52\textwidth]{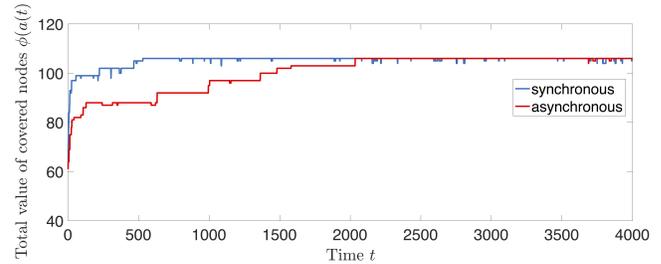}
\caption{{Evolution of the total value of covered nodes under BLLL and its proposed synchronous version in an example.  
}}
\label{fig:plots2}
\end{figure}

Plots in Fig. \ref{fig:plots} show that the proposed synchronization leads to an approximately 4-5 times faster learning than the BLLL in this scenario. Based on the average total values (blue lines) in Fig. \ref{fig:plots}, we observe that the time steps it takes  to reach a total value of 90, 95, and 100 are as follows (synchronous vs. asynchronous): 90 (73 vs. 280), 95 (91 vs. 372), 100 (134 vs. 644). Such an increase in the convergence speed can be explained by the fact that once these 5 agents sufficiently spread out, most of them can move simultaneously under the proposed synchronous algorithm. 
\section{Conclusion}
\label{conc}
We presented a synchronization method for improving the convergence speed of asyncrhonous game-theoretic learning algorithms while maintaining their limiting behavior. We particularly focused on the stochastically stable states, which are widely used for characterizing the limiting behavior of various learning algorithms such as log-linear learning, Metropolis learning, and their variants. In this context, a decentralized random prioritization based method was proposed to enable simultaneous updates by uncoupled agents, who do not affect each other's utility or feasible actions at the current configuration, in each round of the learning process. We theoretically showed the invariance of stochastically stable states under the proposed approach and numerically demonstrated the resulting improvement in convergence speed by considering a coverage control problem in an environment represented as a graph.

\bibliographystyle{IEEEtran}
\bibliography{refer}
\end{document}